\documentclass[10pt,conference]{IEEEtran}
\IEEEoverridecommandlockouts
\usepackage{cite}
\usepackage{amsmath,amssymb,amsfonts}
\usepackage{amsthm}
\usepackage{amsbsy}

\usepackage{bm}
\usepackage{graphicx}
\usepackage{textcomp}
\usepackage{xcolor}
\usepackage{algorithm}
\usepackage{balance}
\usepackage{algpseudocode}
\algblock{Output}{EndOutput}
\algnotext{EndOutput}
\algnewcommand{\Initialize}[1]{%
	\State \textbf{Initialize:}
	\Statex \hspace*{\algorithmicindent}\parbox[t]{.8\linewidth}{\raggedright #1}
}

\usepackage{dblfloatfix}
\usepackage[outdir=./]{epstopdf}
\usepackage{enumitem}
\newlist{steps}{enumerate}{1}
\setlist[steps, 1]{label = Step \arabic*:}
\allowdisplaybreaks
\def\BibTeX{{\rm B\kern-.05em{\sc i\kern-.025em b}\kern-.08em T\kern-.1667em\lower.7ex\hbox{E}\kern-.125emX}}

\newtheorem{proposition}{Proposition} 
\theoremstyle{remark}
\newtheorem{remark}{Remark}

\newcommand{\esp}{\mathbb{E}}
\newcommand{\var}{\operatorname{Var}}

\newcommand{\ttag}{\mathbf{t}}

     \usepackage[nodisplayskipstretch]{setspace}
\begin{document}

\title{Cell-Free Massive MIMO-Based Physical-Layer Authentication   \vspace{-0.3em}

\author{

\IEEEauthorblockN{Isabella W. G. da Silva$^{*}$, Zahra Mobini$^{*\dag}$, Hien Quoc Ngo$^{*}$, and Michail Matthaiou$^{*}$}

\IEEEauthorblockA{$^{*}$Centre for Wireless Innovation (CWI), Queen's University Belfast, U.K.}

\IEEEauthorblockA{$^{\dag}$Department of Electrical and Electronic Engineering, The University of Manchester, Manchester, U.K.}

\IEEEauthorblockA{E-mails: \{iwgdasilva01, hien.ngo, m.matthaiou\}@qub.ac.uk, zahra.mobini@manchester.ac.uk}

\vspace{-0.45cm}
}

\thanks{This work is a contribution by Project REASON, a UK Government funded project under the Future Open Networks Research Challenge (FONRC) sponsored by the Department of Science Innovation and Technology (DSIT). It was also supported by the U.K. Engineering and Physical Sciences Research Council (EPSRC) (grants No. EP/X04047X/1 and EP/X040569/1). The work of I. W. G. da Silva, Z.~Mobini, and  H.~Q.~Ngo was supported by the U.K. Research and Innovation Future Leaders Fellowships under Grant MR/X010635/1, and a research grant from the Department for the Economy Northern Ireland under the US-Ireland R\&D Partnership Programme. The work of M. Matthaiou was supported by the European Research Council (ERC) under the European Union’s Horizon 2020 research and innovation programme (grant agreement No. 101001331).}}

\maketitle

\begin{abstract}
%Tag-based physical-layer authentication (PLA) schemes offer inherent resilience to environmental dynamics and channel variations, positioning them as a valuable complement to traditional upper-layer authentication mechanisms. However, most existing PLA methods are limited to authenticating a single user at a time, which restricts their scalability in multi-user scenarios.
In this paper, we  exploit the cell-free massive multiple-input multiple-output (CF-mMIMO) architecture to design a physical-layer authentication (PLA)  framework that can simultaneously authenticate multiple distributed users across the coverage area.
Our proposed scheme remains effective even in the presence of active adversaries attempting impersonation attacks to disrupt the authentication process. Specifically, we introduce a tag-based PLA  CF-mMIMO system, wherein the access points (APs) first estimate their channels with the legitimate users during an uplink training phase. Subsequently, a unique secret key is generated and securely shared between each user and the APs. We then formulate a hypothesis testing problem and derive a closed-form expression for the probability of detection for each user in the network. Numerical results validate the effectiveness of the proposed approach, demonstrating that it maintains a high detection probability even as the number of users in the system increases.
\end{abstract}

\section{Introduction}\label{sec:Introduction}
%With the ever increasing amount of sensitive and confidential information being transmitted in wireless networks, securing the data becomes crucial for the future sixth-generation (6G) of wireless communications, and traditional cryptography-based techniques, carried out at upper layers may not be sufficient to ensure the security of the network. In this context, physical layer security (PLS) has emerged as a promising candidate to enhance the security of wireless networks 

Physical-layer security (PLS) has attracted significant research interest in recent years, emerging as a promising approach to enhancing the security of wireless networks~\cite{9279294, Mobini:TIFS:2019}. A key operation within PLS is node authentication, which aims to ensure that communicating nodes can accurately and uniquely identify each other. Existing  PLA  schemes can generally be categorized into two types: passive and active. In passive PLA schemes, authentication is performed based on physical-layer characteristics of the received signal, such as hardware impairments~\cite{8920228}, as well as inherent properties of the communication channel, including the channel impulse response and angle of arrival~\cite{10694672}.

Although extensively studied, passive schemes are often sensitive to external factors such as temperature fluctuations and channel variability, which can limit their practical applicability. In contrast, active schemes involve a mutual agreement between the transmitter and receiver on a proof of authentication—commonly referred to as a tag—which is embedded within the transmitted signal and extracted by the receiver during data transmission~\cite{10005830}. Tag-based PLA schemes offer various advantages, including low overhead, reduced complexity, and enhanced security. The randomness introduced by channel fading and noise increases the uncertainty for potential attackers, thereby improving robustness~\cite{10031268}. For example, in~\cite{10005830}, a tag-based PLA scheme was introduced for a reconfigurable intelligent surface (RIS)-assisted communication system, where the tag is constructed using the channel characteristics, background noise, and a random signal to thwart impersonation attacks. In addition,~\cite{10031268} addressed the challenge of authenticating and classifying transmitters in a multi-user scenario relying on orthogonal transmissions.  

 A common assumption in the existing literature is the presence of a  single-antenna receiver  and either a single transmitter or multiple transmitters availing of multi-carrier transmission. However, multi-carrier transmission becomes increasingly inefficient in terms of resource allocation as the number of users grows. The requirement to assign orthogonal subcarriers or time slots to each user leads to scalability issues and reduced spectral efficiency (SE).  On the other hand,  studying active PLA schemes with  multiple-input multiple-output (MIMO)/massive MIMO   receivers remains largely unexplored.  PLA with a MIMO receiver was investigated in~\cite{9672766} and~\cite{5740593}, where the authors proposed a fingerprint-embedding authentication framework and demonstrated that authentication performance can be enhanced when precoding is applied to both the data and the fingerprint. However, a key limitation of these studies is the assumption of a single receiver and the restriction that only one user—either a legitimate user or an attacker—can communicate with the receiver during a given transmission block.
 %This constraints limit  the applicability of the proposed schemes in more practical multi-user scenarios.  

To address this research gap, we propose an active PLA scheme capable of simultaneously authenticating multiple users in a distributed wireless environment. Our approach is built upon the CF-mMIMO architecture, which has emerged as a promising technology for next-generation wireless networks~\cite{Mohammadi:PROC.2024}.
%CF-mMIMO eliminates the traditional notion of cell boundaries by deploying a large number of distributed access points (APs) connected to a central processing unit (CPU), enabling coherent joint transmission and reception. These APs collaboratively serve a comparatively smaller set of user terminals, offering significant gains in spectral efficiency, reliability, energy efficiency, and coverage uniformity~\cite{7827017}.
In the proposed scenario, multiple legitimate users aim to be authenticated by multiple distributed APs connected to a central processing unit (CPU), while a malicious attacker attempts to compromise the process by injecting spoofed signals to impersonate legitimate users. Our scheme enables concurrent authentication of multiple users without relying on orthogonal resources or massive antenna arrays at a centralized receiver. By leveraging the spatial diversity and decentralized structure of CF-mMIMO, the proposed scheme enhances both the scalability and robustness of physical-layer authentication in the presence of sophisticated adversaries. The main contributions of this work are as follows:
\begin{itemize}
    \item We propose a CF-mMIMO tag-based PLA scenario, in which the APs first estimate their channels with the legitimate users using pilot signals from the uplink training phase. Subsequently, a unique secret key is generated and securely shared between each user and the APs. %Unlike previous works (e.g.,~\cite{5740593, 10031268}), our approach considers simultaneous transmissions from both legitimate users and the attacker.   
    \item We formulate the authentication problem as a hypothesis testing framework and derive closed-form expressions for the probability of false alarm (PFA) and probability of detection (PD) for each legitimate user in the CF-mMIMO network. 
    \item Numerical results demonstrate that the proposed CF-mMIMO PLA  system   effectively enhances  the PD  across various scenarios. In particular, a PD performance of over $90\%$ is achieved when there are 8 APs in the system.
\end{itemize}

%Notwithstanding the benefits of the CF-mMIMO architecture, it introduces additional challenges in ensuring a secure communications 

\textit{Notation. } Throughout this paper, bold upper-case letters denote matrices whereas bold lower-case letters denote vectors; $(\cdot)^T$ and $(\cdot)^H$ stand for the matrix transpose and Hermitian transpose, respectively; $\mathrm{Tr}(\cdot)$ is the trace operator; $\|\cdot\|$ and $|\cdot|$ are the Euclidean-norm and the absolute value operator; $\esp\{\cdot\}$ is the expectation operator, and $\var(a)\triangleq\esp\left\{|a - \esp\{a\}|^2\right\}$; $\mathbb{R}\{\cdot\}$ represents a real-value extractor. Finally, a circular symmetric complex Gaussian vector $\mathbf{x}$ with covariance matrix $\mathbf{C}$ is denoted by $\mathbf{x}\sim\mathcal{CN}(\mathbf{0}, \mathbf{C})$.

\section{System Model}
%\vspace{-0.1cm}
\begin{figure}[t]
    \centering
    \includegraphics[scale=0.45]{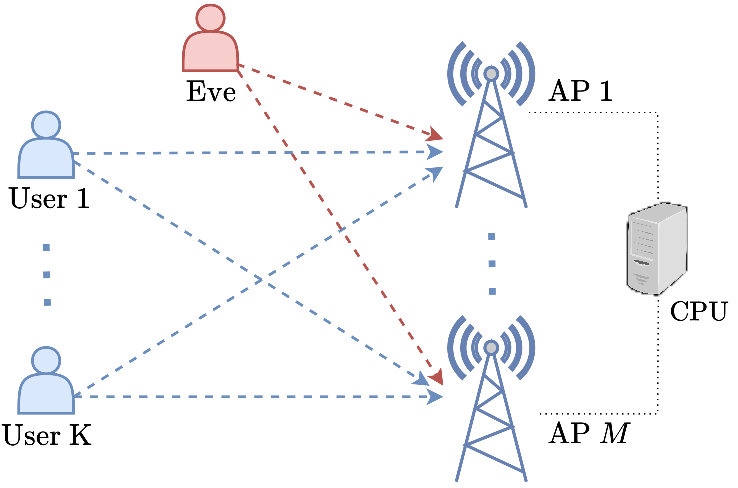}
    %\vspace{-0.5cm}
    \caption{CF-mMIMO tag-based PLA system model.}
      %\vspace{-0.5cm}
    \label{fig:sysmodel}
\end{figure}
As illustrated in Fig.~\ref{fig:sysmodel}, we consider a CF-mMIMO tag-based PLA  scenario in which $K$ legitimate users communicate with $M$ APs in the presence of an attacker, denoted as Eve. The APs are equipped with $N$ antennas, while both the legitimate users and Eve are assumed to have a single antenna. Moreover, all the nodes operate in half-duplex mode. All APs connect to a CPU which aims to authenticate the messages coming from the legitimate users. A PLA scheme based on a shared secret key between the APs and users is considered. Specifically, the $k$th user and the APs generate a secret key $\mathbf{l}_k$, which is unique for each user and unknown to Eve. The secret key is then employed by the users to generate a proof of authentication, named tag. Each user superimposes the created tag onto its message signal, and transmits the composition of the message plus tag to the APs.  We assume that Eve is a malicious user aware of the authentication process but with no knowledge of the secret key, thus Eve aims to disrupt the authentication process by impersonating the legitimate users, transmitting spoofing messages to the APs, hoping that the CPU mistakenly accepts it as authentic. 

%Different from previous works in the literature of PLA, we consider multiple users communicating simultaneously to the APs. For instance, in~\cite{10031268}, the authors considered multiple transmitters but only one of them would make the transmission at a certain transmission block

To reflect a more practical scenario, we consider an uplink training phase between the legitimate users and the APs. During this phase, the users transmit uplink pilot signals, allowing the APs to estimate their respective channels with each user. Once the channel estimation is complete, the authentication scheme is initiated using the channels estimated during the training phase.. The details of the uplink training procedure are provided in the following subsection.

%First, to account for a more practical scenario, we consider the uplink training stage between the users and APs. During the uplink training, the users~transmits uplink pilots to the APs, enabling them to estimate its channel with each of the users. Note that since the uplink training stage is prior to authentication process, it is assumed that Eve also transmits pilots to the APs.  After the channel estimation phase, the authentication scheme is established. Accordingly, the details of the uplink training is provided in the next subsection. 

\subsection{Uplink Training}\label{sec:uptraining}
We assume that the $K$ users simultaneously transmit pilot sequences to the APs, with the pilot sequence transmitted by the $k$th user being $\bm{\varphi}_k \in \mathbb{C}^{\tau_{\text{p}} \times 1}$, where $\tau_{\text{p}}$ is the pilot sequence length. The pilot sequences are considered pairwisely orthonormal, i.e., $\bm{\varphi}_k^H\bm{\varphi}_{k'}=0$ if $k\neq k'$, with $\|\bm{\varphi}_k\|^2=1$, $\forall k = 1, \dots, K$. Thus, it is required that $\tau_{\text{p}}\geq K$. Hence, the received pilot matrix at the $m$th AP is given by
\begin{align}
    \mathbf{Y}_{\text{u},m} = \sqrt{\tau_{\text{p}}}\sum_{k=1}^{K} \rho_k\mathbf{g}_{mk}\bm{\varphi}_{k}^H + \bm{\Omega}_{\text{u},m},
\end{align}
where $\rho_k$ is the maximum normalized transmit power for the $k$th user, and $\bm{\Omega}_{\text{u},m} \in \mathbb{C}^{N \times \tau_{\text{p}}}$ is the noise matrix at the $m$th AP, assumed to have independent and identically distributed (i.i.d.) $\mathcal{CN}(0,1)$ entries. In addition, $\mathbf{g}_{mk} \in \mathbb{C}^{N \times 1}$ is the channel vector from the $k$th user to the $m$th AP, modeled as
\begin{align}\label{eq:gmodel}
    \mathbf{g}_{mk} = \sqrt{\beta_{mk}}\mathbf{h}_{mk},
\end{align}
where $\beta_{mk}$ is the large-scale fading coefficient, while $\mathbf{h}_{mk}$ is the small-scale fading vector, with i.i.d. $\mathcal{CN}(0,1)$ entries. Next, the received signal $\mathbf{Y}_{\text{u},m}$ is projected onto  $\bm{\varphi}_k$ to estimate the channel vector from user $k$ to the $m$th AP as 
\begin{align}\label{eq:tildey}
     \Tilde{\mathbf{y}}_{mk} &= \mathbf{Y}_{\text{u},m}\bm{\varphi}_{k} \nonumber\\ 
     &= \sqrt{\tau_{\text{p}}\rho_k}\mathbf{g}_{mk}+ \bm{\Omega}_{\text{u},m}\bm{\varphi}_{k}.
\end{align}

Given~\eqref{eq:tildey}, the estimate of $\mathbf{g}_{mk}$ can be attained via a linear minimum mean-squared error (MMSE) approach as
\begin{align}
    \hat{\mathbf{g}}_{mk} &= \esp\{\mathbf{g}_{mk}\Tilde{\mathbf{y}}_{mk}^H\}\left(\esp\left\{\Tilde{\mathbf{y}}_{mk}\Tilde{\mathbf{y}}_{mk}^H\right\}\right)^{-1}\Tilde{\mathbf{y}}_{mk}\nonumber\\
    &=\frac{\sqrt{\tau_{\text{p}}\rho_k}\beta_{mk}}{\tau_{\text{p}}\rho_k\beta_{mk}+1}\Tilde{\mathbf{y}}_{mk}.
\end{align}
Thus, $\hat{\mathbf{g}}_{mk}$ has i.i.d. $\mathcal{CN}\left(0,\frac{\tau_{\text{p}}\rho_k\beta_{mk}^2}{\tau_{\text{p}}\rho_k\beta_{mk}+1}\right)$ elements.

\section{Authentication Framework}
In this section, we present the proposed  CF-mMIMO PLA scheme to authenticate the users  in the presence of an active attacker. First, we provide details on the authentication tag generation process employed by the users and the transmitted tagged signal. Subsequently, we describe the CPU processing to recover the tags and the authentication hypothesis test.   

%a description about the authentication tag generation and signal transmission process followed by the users is provided. Subsequently, We describe the transmission process followed by the users in generating the authentication tag, and embedding  on the transmitted signal by the users.   
\subsection{Transmission Phase}
 After the uplink training phase, the legitimate users and APs agree on an authentication scheme based on a shared secret key. The $k$th user generates and shares its secret key $\mathbf{l}_k$ with the APs. The secret keys from all of the users are assumed to be independent from each other and unknown to Eve~\cite{8546764}. Next, during the data transmission phase, if user $k$ intends to be authenticated by the CPU, a proof of authentication should be transmitted along with the message signal. This proof of authentication, commonly named as tag, reflects knowledge of the secret key. 

 The data is considered to be transmitted in blocks, with length $L$. Let $\mathbf{s}_{k} = [s_{k,1}, \dots, s_{k,L}] \in \mathbb{C}^{L \times 1}$ denote the message signal from the $k$th user in a particular transmission block.  
 Then, the tagged signal transmitted by the $k$th user is
 \begin{align}\label{eq:taggedsignal}
     \mathbf{x}_{k} = \rho_{\text{s}} \mathbf{s}_{k} + \rho_{\text{t}} \ttag_{k},
 \end{align}
 where $\ttag_{k} = [t_{k,1}, \dots, t_{k,L}] \in \mathbb{C}^{L \times 1}$ is the authentication tag, written as
 \begin{align}\label{eq:tag}
     \ttag_{k} = g(\mathbf{s}_{k}, \mathbf{l}_k),     
 \end{align}
 where $g(\cdot)$ is a secure hash function~\cite{bookCrypto}. We assume that the message signal and the authentication tag satisfy $\esp\{\|\mathbf{s}_k\|^2\}=\esp\{\|\ttag_k\|^2\}=L$ and  $\esp\{s_{k,l}\} = \esp\{t_{k,l}\} = 0, \forall k \in \{1, \dots, K\} , \forall l \in \{1, \dots, L\}$,  while $\esp\{\|\mathbf{x}_k\|^2\}=L$.
 We also consider that the users generate independent authentication tags, $\esp\{\ttag_k^H\ttag_{k'}\}=0, \forall k' \neq k, k' \in \{1, \dots, K\}$, and that the message signals are not correlated to the tags~\cite{4451099,10031268}, i.e., $\esp\{\ttag_k^H\mathbf{s}_{k}\}=0$. Moreover, $\rho_{\text{s}}$ and $\rho_{\text{t}}$ are the power allocation coefficients of the message and of the tag signal, respectively, which are carefully chosen to satisfy $\rho_{\text{t}} \ll  \rho_{\text{s}}$ and {$\rho_{\text{s}}^2 + \rho_{\text{t}}^2 = 1$}. Consequently, if $\rho_{\text{s}}=1$, it indicates that user $k$ is not transmitting the authentication tag during the transmission block.  
  
  At the same time, Eve conducts an impersonation attack to disrupt the authentication process. Since Eve is assumed to be aware of all the authentication processes except for the secret key~\cite{8546764}, its transmitted signal is described as 
 \begin{align}
     \mathbf{x}_{\text{e}} = \rho_{\text{s}} \mathbf{s}_{\text{e}} + \rho_{\text{t}} \ttag_{\text{e}},
 \end{align}
 where $\mathbf{x}_{\text{e}}$ is the signal transmitted by Eve, with $\mathbf{s}_{\text{e}} = [s_{\text{e},1}, \dots, s_{\text{e},L}] \in \mathbb{C}^{L \times 1}$ being the message signal from Eve. The false authentication tag is given by $\ttag_{\text{e}} = g(\mathbf{s}_{\text{e}}, \mathbf{l}_{\text{e}})$, where $\mathbf{l}_{\text{e}}$ is a random key generated by Eve. It is further assumed that the message signal and authentication tag generated by Eve should also satisfy $\esp\{\|\mathbf{s}_{\text{e}}\|^2\}=\esp\{\|\ttag_{\text{e}}\|^2\}=L$ and $\esp\{s_{\text{e},l}\} = \esp\{t_{\text{e},l}\} = 0, \forall l \in \{1, \dots, L\}$. Also, we consider that $\esp\{\|\mathbf{x}_{\text{e}}\|^2\}=L$, and that the false authentication tag is not correlated to the message signal $\esp\{\ttag_{\text{e}}^H\mathbf{s}_{\text{e}}\}=0$. Hence, the $N \times L$ received signal at the $m$th AP is given by
 \begin{align}\label{eq:sigm}
     \mathbf{\bar{Y}}_{m} = \sum_{k=1}^{K}\sqrt{\rho_k}\mathbf{g}_{mk}\mathbf{x}_{k}^H+\sqrt{\rho_{\text{e}}}\mathbf{g}_{m\text{e}}\mathbf{x}_{\text{e}}^H+\bm{\bar{\Omega}}_{m},
 \end{align}
 where $\bm{\bar{\Omega}}_{m}\in\mathbb{C}^{N\times L}$ is the noise matrix at the $m$th AP, which is assumed to have i.i.d. $\mathcal{CN}(0,1)$ entries. In addition, $\mathbf{g}_{m\text{e}} \in \mathbb{C}^{N \times 1}$ is the channel vector between Eve and the $m$th AP, modeled as 
 \begin{align}\label{eq:gemodel}
    \mathbf{g}_{m\text{e}} = \sqrt{\beta_{m\text{e}}}\mathbf{h}_{m\text{e}},
\end{align}
 where $\beta_{m\text{e}}$ is corresponding the large-scale fading coefficient, and $\mathbf{h}_{m\text{e}}$ is the associated small-scale fading vector, assumed to be $\mathcal{CN}(\mathbf{0},\mathbf{I}_N)$ distributed. Moreover, $\rho_{\text{e}}$ is the maximum normalized transmit power at Eve.
\subsection{Signal Detection and Authentication Phase}
 To detect $\mathbf{x}_{k}$, the $m$th AP uses the channel estimates acquired in the uplink training phase described previously in Sec.~\ref{sec:uptraining}, $\hat{\mathbf{G}}_{m}=[\hat{\mathbf{g}}_{m1}, \dots, \hat{\mathbf{g}}_{mK}]$, to combine its received signal. A widely-used local zero-forcing (ZF) combining technique is deployed. This requires $N \geq K$. More precisely, an $N \times K$ ZF combining matrix is designed as
 \begin{align}
     \mathbf{W}_m = \hat{\mathbf{G}}_{m}\left(\hat{\mathbf{G}}_{m}^{H}\hat{\mathbf{G}}_{m}\right)^{-1}.
 \end{align}
Thus, the aggregated received signal at the CPU is given by
 \begin{align}\label{eq:zc}
     \mathbf{Z}_{\mathrm{c}}^H &= \sum_{m=1}^{M}\mathbf{W}_m^{H}\mathbf{\bar{Y}}_{m}\nonumber\\
     &=\sum_{m=1}^{M}\left(\mathbf{W}_m^{H}\left(\sum_{k=1}^{K}\sqrt{\rho_k}\mathbf{g}_{mk}\mathbf{x}_{k}^H+\sqrt{\rho_{\text{e}}}\mathbf{g}_{m\text{e}}\mathbf{x}_{\text{e}}^H+\bm{\bar{\Omega}}_{m}\right)\right).
 \end{align}
The $k$th row of $\mathbf{Z}_{\mathrm{c}}^H$ is used to detect $\mathbf{x}_k$. From \eqref{eq:zc}, the $k$th row of $\mathbf{Z}_{\mathrm{c}}^H$ is given by
\begin{align}\mathbf{z}_{\text{c},k}^H=\sum_{m=1}^{M}\left(\mathbf{d}_{mk}^H+\mathbf{i}_{mk}^H+\mathbf{n}_{mk}^H\right),
\end{align}
where
%\vspace{-0.2em}
\begin{align}
    \mathbf{d}_{mk}^H &\triangleq \sqrt{\rho_k}\mathbf{w}_{mk}^H\mathbf{g}_{mk}\mathbf{x}_{k}^H,\label{eq:dc}\\
    \mathbf{i}_{mk}^H &\triangleq \sum_{k' \neq k}^{K}\sqrt{\rho_{k'}}\mathbf{w}_{mk}^H\mathbf{g}_{mk'}\mathbf{x}_{k'}^H,\\
    \mathbf{n}_{mk}^H &\triangleq \sqrt{\rho_{\text{e}}}\mathbf{w}_{mk}^H\mathbf{g}_{m\text{e}}\mathbf{x}_{\text{e}}^H+\mathbf{w}_{mk}^H\bm{\bar{\Omega}}_{m},\label{eq:nc}
\end{align}
where $\mathbf{w}_{mk}$ is the $k$th column of $\mathbf{W}_m$. Also, $\mathbf{d}_{k}$, $\mathbf{i}_{k}$ and $\mathbf{n}_{k}$ stand  for the desired signal, the multi-user interference, and the noise corresponding to user $k$, respectively. Note that, since the APs do not have prior knowledge about the presence of Eve in the communication process, her signal is regarded as noise. Furthermore, given that $\rho_{\text{t}} \ll  \rho_{\text{s}}$, the estimation of the message signal from the $k$th user, $\hat{\mathbf{s}}_k$, can be obtained from $\mathbf{z}_{\text{c},k}$ with a tolerable bit error rate (BER), while ignoring the presence of the authentication tag~\cite{7120016,10005830}. 

After recovering the messages, the CPU must decide if they are   authentic or not. To this end, the CPU uses the estimated message $\hat{\mathbf{s}}_k$ and the secret key corresponding to the $k$th user, $\mathbf{l}_k$, to construct an expected authentication tag, $\Tilde{\ttag}_k$ as
\begin{align}
      \Tilde{\ttag}_k = g(\hat{\mathbf{s}}_k, \mathbf{l}_k). 
\end{align}
After generating $\Tilde{\ttag}_k$, it is also necessary to recover the received tag from the residual signal, denoted by  $\mathbf{r}_k$, as  
\begin{align}
    \mathbf{r}_k = \frac{1}{\rho_{\text{t}}}\left(\frac{\mathbf{z}_{\text{c},k}}{\sqrt{\rho_k}} - \rho_{\text{s}}\hat{\mathbf{s}}_k\right).
\end{align}
Then, the authentication problem can be summarized as a binary hypothesis testing problem for each  user $k$. In particular, the hypothesis are expressed as
\begin{align}\label{eq:autdes}
    \left\{\begin{array}{l}
\mathcal{H}_0: \Tilde{\ttag}_k \text{ is not present in } \mathbf{r}_k, \\
\mathcal{H}_1: \Tilde{\ttag}_k \text{ is present in } \mathbf{r}_k,
\end{array}\right. 
\end{align}
where $\mathcal{H}_1$ represents that the correct tag exists in the
received signal, and user $k$ is being correctly authenticated, whereas $\mathcal{H}_0$ denotes that $\mathbf{r}_k$ does not contain the correct tag. To perform the hypothesis test, the CPU performs a threshold test between $\mathcal{H}_0$ and $\mathcal{H}_1$ for each user $k$ to make the authentication decision, which is written as
\begin{align}\label{eq:threshold}
    \lambda_k \stackrel{\mathcal{H}_1}{\underset{\mathcal{H}_0}{\gtrless}} \theta_k,
\end{align}
where $\lambda_k$ and $\theta_k$ are the test statistic and the decision threshold of the $k$th user, respectively. Here,  $\lambda_k$ is constructed by match-filtering the residual signal $\mathbf{r}_k$ with $\Tilde{\ttag}_k$, that is
\begin{align}\label{eq:test}
    \lambda_k =  \mathbb{R}\{\Tilde{\ttag}_k^H\mathbf{r}_k\}.
\end{align}
As discussed in~\cite{4451099}, if the employed hash function $g(\cdot)$ is robust against errors, $\Tilde{\ttag}_k$ can be correctly generated even if the recovered message contains some errors. Therefore, it is feasible to assume perfect message recovery ($\hat{\mathbf{s}}_k = \mathbf{s}_k$), and tag estimation ($\Tilde{\ttag}_k = \ttag_k$) for the hypothesis testing. For notation simplicity, let us define
\begin{align}
    a_{k,k} &\triangleq \sqrt{\rho_{k}}\sum_{m=1}^{M}\mathbf{w}_{mk}^H\mathbf{g}_{mk},\\
    a_{k,k'} &\triangleq \sqrt{\rho_{k'}}\sum_{m=1}^{M}\mathbf{w}_{mk}^H\mathbf{g}_{mk'}, \forall k' \neq k,\\
    a_{k,\text{e}} &\triangleq \sqrt{\rho_{\text{e}}}\sum_{m=1}^{M}\mathbf{w}_{mk}^H\mathbf{g}_{m\text{e}}.
\end{align}
We derive the expressions for the test statistic when the $k$th user transmits only the message signal ($\lambda_k | \mathcal{H}_0$), and when a tagged signal is received from the $k$th user ($\lambda_k | \mathcal{H}_1$), which are later employed in Section~\ref{sec:perf} to compute the PFA and PD of a certain user $k$. Accordingly, when the $k$th user transmits only the message signal ($\rho_{\text{s}} = 1$), the test statistic is expressed as
\begin{align}\label{eq:lambdakh0}
   \lambda_k | \mathcal{H}_0 &=\ttag_k^H\mathbf{r}_k \nonumber\\ 
   &= \frac{(a_{k,k}-\sqrt{\rho_k})\rho_{\text{s}}}{\sqrt{\rho_k}\rho_{\text{t}}}\ttag_k^H\mathbf{s}_k + \ttag_k^H\sum_{k' \neq k}^K\frac{a_{k,k'}}{\sqrt{\rho_k}\rho_{\text{t}}}\mathbf{x}_{k'}\nonumber\\
    &~~+\frac{a_{k,\text{e}}}{\sqrt{\rho_k}\rho_{\text{t}}}\ttag_k^H\mathbf{x}_{\text{e}} + \frac{1}{\sqrt{\rho_k}\rho_{\text{t}}}\ttag_k^H\mathbf{n}_k,
\end{align} 
where $\mathbf{n}_k \triangleq \sum_{m=1}^M \left(\mathbf{w}_{mk}^H\bm{\bar{\Omega}}_{m}\right)^H$. On the other hand, when the tagged signal is transmitted, we obtain
\begin{align}\label{eq:lambdakh1}
    \lambda_k | \mathcal{H}_1 &= \frac{a_{k,k}}{\sqrt{\rho_k}}\|\ttag_k\|^2 +\frac{(a_{k,k}-\sqrt{\rho_k})\rho_{\text{s}}}{\sqrt{\rho_k}\rho_{\text{t}}}\ttag_k^H\mathbf{s}_k \nonumber\\
    &~~+ \ttag_k^H\sum_{k' \neq k}^K\frac{a_{k,k'}}{\sqrt{\rho_k}\rho_{\text{t}}}\mathbf{x}_{k'} +\frac{a_{k,\text{e}}}{\sqrt{\rho_k}\rho_{\text{t}}}\ttag_k^H\mathbf{x}_{\text{e}} + \frac{1}{\sqrt{\rho_k}\rho_{\text{t}}}\ttag_k^H\mathbf{n}_k.
\end{align}

\section{Performance Analysis}\label{sec:perf}
In this section, we aim to evaluate the performance of our proposed PLA scheme. Based on the threshold test in \eqref{eq:threshold}, if the CPU accepts $\mathcal{H}_1$ when $\mathcal{H}_0$ is true for the $k$th user, then a false alarm has occurred. Meanwhile, if the CPU accepts $\mathcal{H}_1$ when $\mathcal{H}_1$ is true, a correct detection has happened.  Accordingly, we derive the closed-form expressions for the PFA and PD for the $k$th legitimate user, and we also determine the optimal decision threshold for each user $k$, $\theta^*_k$. The optimal decision threshold is computed according to the Neyman-Pearson criterion~\cite{850674}. Thus, $\theta^*_k$ is set to maximize the PD of the $k$th user for a fixed PFA rate, $p_{\mathrm{FA},k}$. The closed-form expressions for the PFA, PD and for the decision threshold are presented in the next proposition. 
\begin{proposition}\label{prop:cf}
    Based on the authentication decision in~\eqref{eq:threshold}, the closed-form expressions for the PFA, the PD, and the optimal decision threshold of a certain user $k$ in a CF-mMIMO system in the presence of an attacker performing an impersonation attack are given, respectively, by
    \begin{align}
        \mathrm{P}_{\mathrm{FA}, k}&=Q\left(\frac{\theta_k}{\sqrt{L \xi_k}}\right),\label{eq:pfakf}\\
        \mathrm{P}_{\mathrm{D}, k}&=Q\left(\frac{\theta_k - ML}{\sqrt{L\xi_k}}\right),\label{eq:pdkf}\\
    \theta_k^* &= Q^{-1}(p_{\mathrm{FA},k})\sqrt{L \xi_k},\label{eq:opttheta}
    \end{align}
    where $Q(\cdot)$ represents the Q-function. Moreover, $\xi_k$ is given by
    \begin{align}
        \xi_k = &\frac{\rho_{\text{s}}^2}{\rho_k\rho_{\text{t}}^2}(\sigma_{k,k}^2+(\sqrt{\rho_k}M-1)^2) + \left(\frac{1}{\rho_k} + \frac{\rho_{\text{s}}^2}{\rho_k\rho_{\text{t}}^2}\right)\nonumber\\
        &\times\left(\sum_{k' \neq k}^K\sigma_{k,k'}^2 + \sigma_{k,e}^2\right)+ \frac{1}{\rho_k\rho_{\text{t}}^2}\mathrm{Tr}(\mathbf{C}_{nk}),
    \end{align}
    where $\sigma_{k,k}^2$, $\sigma_{k,k'}^2$ and $\sigma_{k,e}^2$ denote the variances of $a_{k,k}$, $a_{k,k'}$ and $a_{k,e}$, respectively. Also, $\mathbf{C}_{nk}$ is the covariance of $\mathbf{n}_k$.
\end{proposition}
\begin{proof}
    The proof is provided in Appendix A.
\end{proof}
\begin{remark}
 The versatility of the results in Proposition~\ref{prop:cf}  extends to various scenarios. For instance, the hypothesis testing procedure, as well as the corresponding $\mathrm{P}_{\mathrm{FA}, k}$ and  $\mathrm{P}_{\mathrm{D}, k}$ coincide with those previously reported in the literature in the special case of perfect CSI availability and single-antenna AP acting as the receiver and a single user (i.e., $M = N = K = 1$)~\cite{4451099, 8546764}.
\end{remark}
 
\section{Numerical Results}\label{sec:Results}
In this section, the performance of the proposed CF-mMIMO PLA system is assessed in terms of the PD with an optimal decision threshold for different PFA levels. The APs, legitimate users, and Eve are randomly distributed within a $1\times 1$ km$^2$ area with wrapped-around edges to avoid boundary effects. The channel bandwidth is $B=20$ MHz and the pilot length is $\tau_{\text{p}} = 20$. The maximum transmit power for each user and for Eve is set at 100 mW. The corresponding normalized transmit powers, $\rho_k$ and $\rho_{\text{e}}$, can be obtained by dividing the maximum transmit power by the noise power, given by
\begin{align}
    \text{noise power} = B \times k_B \times T_0 \times \text{noise figure }  (\mathrm{W}),
\end{align}
where $k_B = 1,381 \times 10^{-23}$ (Joule per Kelvin) is the Boltzmann constant, $T_0 = 290$ (Kelvin) is the noise temperature, and the noise figure is 9 dB. The large-scale fading coefficient $\beta_{mk}$ is computed as
\begin{align}
    \beta_{mk} = \mathrm{PL}_{mk} 10^{\frac{\sigma_{\mathrm{sh}}z_{\mathrm{mk}}}{10}}, 
\end{align}
where $\mathrm{PL}_{mk}$ represents the path loss, which is computed as in \cite[Eq. (52)]{7827017}, and $10^{\frac{\sigma_{\mathrm{sh}}z_{mk}}{10}}$ represents the shadow fading with the standard deviation $\sigma_{\mathrm{sh}} = 8$ dB, and $z_{mk} \sim \mathcal{N}(0, 1)$. Unless specified otherwise, the length of the transmission block is set as $L=256$, the power allocated for the message signal is $\rho_{\text{s}}=0.95$, and the allowable PFA by the PLA scheme is $p_{\mathrm{FA},k}=0.01$, while $\theta_k^*$ and $\mathrm{P}_{\mathrm{D}, k}$ are computed according to Proposition~\ref{prop:cf}.  

\begin{figure}[t]
    \centering
    \includegraphics[width=0.45\textwidth]{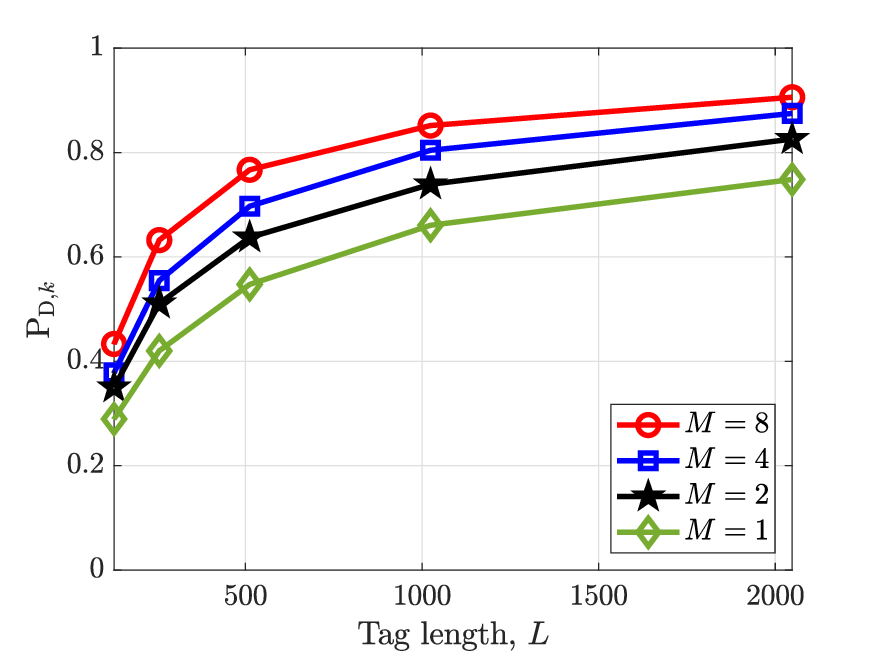}
    \caption{$\mathrm{P}_{\mathrm{D}, k}$ versus the tag length, $L$,  with $K=4$ and $N=10$.}
    \label{fig:PDkvsL}
\end{figure}

As discussed in~\cite{4451099}, a robust tag-based PLA scheme requires sufficient tag energy to enable reliable detection. Since the condition $\rho_{\text{t}} \ll \rho_{\text{s}}$ must be met to ensure that $\hat{\mathbf{s}}_k$ can be accurately recovered from $\hat{\mathbf{x}}_k$ with an acceptable BER, the only viable way to increase the tag energy is by extending the tag length. To examine how the tag length affects the PD performance, Fig.~\ref{fig:PDkvsL} shows $\mathrm{P}_{\mathrm{D}, k}$ as a function of the tag length for various numbers of APs, $ M$, where $K = 4$ and $N = 10$. It is observed that a longer tag improves authentication performance.
However, since the tag is assumed to span the entire transmission block, the available tag energy is limited by the channel coherence time—a parameter that depends on practical propagation environments, and hence, is difficult to control. To overcome this, previous PLA schemes have proposed coding tags across multiple blocks (e.g., requiring detection in two out of four blocks)~\cite{4451099,5740593}. While this method can improve detection, it also increases the risk of falsely accepting malicious messages from Eve. In this regard, our proposed CF-mMIMO infrastructure offers a way to strengthen the PLA scheme, as increasing the number of APs leads to a better PD performance for a fixed tag length $L$. Specifically, a performance gain of about $20\% $ is observed as $M$ increases from $1$ to $8$, regardless of the value of $L$.

\begin{figure}[t]
    \centering
    \includegraphics[width=0.45\textwidth]{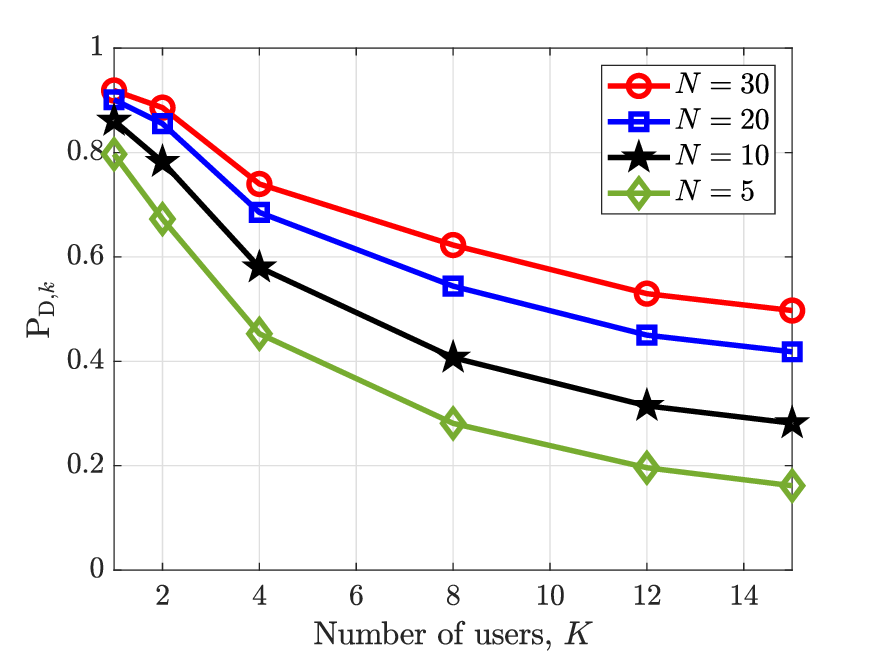}
    \caption{$\mathrm{P}_{\mathrm{D}, k}$ versus the number of users, $K$, with $M=5$.}
    \label{fig:pdkvsk}
\end{figure}

Figure~\ref{fig:pdkvsk} illustrates the detection probability, $\mathrm{P}_{\mathrm{D}, k}$, for a given user $k$ as the total number of users $K$ increases, across different numbers of antennas per  AP, $N$. As $K$ increases, the probability of detecting a specific user decreases due to increased multi-user interference. However, the CF-mMIMO PLA system notably benefits from the array gain associated with increasing the number of antennas per AP, particularly as the number of users grows. Specifically, for $K=1$, increasing $N$ from $5$ to $30$ yields a $10\%$ improvement in $\mathrm{P}_{\mathrm{D}, k}$. For $K=15$, the performance gain rises to $34\%$ over the same antenna range, with $\mathrm{P}_{\mathrm{D}, k}$ still exceeding $50\%$ in this case. These results reinforce the findings from Fig.~\ref{fig:PDkvsL}, highlighting that both micro- and macro-diversity gains provided by CF-mMIMO systems substantially enhance the authentication performance and sustain a satisfactory PD performance even as the number of users increases.

%\comm{I believe that as the number of users increase, we can reach an asymptotic behavior, I'm running the results from $K=20$ to check this and I'll add that to the analysis}.

\begin{figure}[t]
    \centering
    \includegraphics[width=0.45\textwidth]{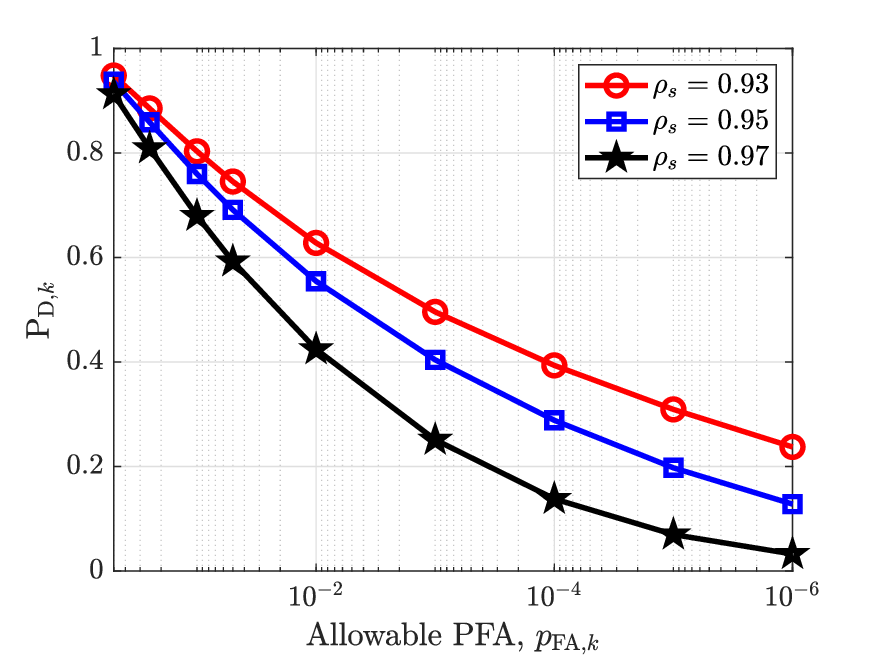}
    \caption{$\mathrm{P}_{\mathrm{D}, k}$ versus the allowable PFA, $p_{\mathrm{FA},k}$,  with $K=4$, $M= 5$, and $N=10$.}
    \label{fig:PDkvspfa}
\end{figure}

Finally, Fig.~\ref{fig:PDkvspfa} shows the $\mathrm{P}_{\mathrm{D}, k}$ versus the allowable PFA for user $k$, $p_{\mathrm{FA},k}$,  for different values of $\rho_{\text{s}}$.  As the allowable PFA increases, the $\mathrm{P}_{\mathrm{D}, k}$ also increases. Moreover, when a stricter power allocation is used for the message signal, the $\mathrm{P}_{\mathrm{D}, k}$ decreases, which is expected due to reduced tag energy. It is important to point out that we assume all users and Eve employ the same power allocation coefficients $\rho_{\text{s}}$ and $\rho_{\text{t}}$. However, in a heterogeneous network, users may have varying quality-of-service requirements, which could also be leveraged to improve the user-specific detection and system performance. Further performance gains may be possible by individually optimizing the power allocation coefficients for each user. We identify this direction as a key focus for future work.

\section{Conclusion}\label{sec:Conclusions}
%\vspace{-1.5mm}
We proposed a CF-mMIMO PLA framework for authenticating multiple users in the presence of an attacker that launches impersonation attacks. In our system, the APs estimate the channels of legitimate users during the uplink training phase. Then, during the data transmission phase, the users send authentication tags generated from a shared secret key prior agreed with the APs.   We formulated a hypothesis testing problem and derived a closed-form expression for the PD of each user. Our numerical results showcased that the proposed CF-mMIMO PLA framework is effective in improving the robustness of authentication, by taking advantage of the micro- and macro-diversity gains achieved with CF-mMIMO. Furthermore, the proposed methodology may achieve even better performance by adapting to user-specific requirements, such as customized power allocation.
%\vspace{-1.5mm}
\section*{Appendix A \\Proof of Proposition~1}
The PFA and PD for the $k$th user are given  by
\begin{align}
    \mathrm{P}_{\mathrm{FA}, k} &= \Pr(\lambda_k > \theta_k | \mathcal{H}_0)\label{eq:pfak},\\
    \mathrm{P}_{\mathrm{D}, k} &= \Pr(\lambda_k > \theta_k | \mathcal{H}_1)\label{eq:pdk},
\end{align}
respectively.
From~\eqref{eq:test}, note that $\lambda_k$ is given by the sum of several independent terms. Thus, by the central limit theorem, $\lambda_k$ is approximately Gaussian distributed, especially as $M$, $N$ and/or $K$ increases~\cite{5740593}. Thus, \eqref{eq:pfak} and \eqref{eq:pdk} can be rewritten, respectively, as
\begin{align}
    \mathrm{P}_{\mathrm{FA}, k}&=Q\left(\frac{\theta_k-\esp\left[\lambda_k|\mathcal{H}_0\right]}{\sigma_{\lambda_k \mid \mathcal{H}_0}}\right),\label{eq:Pfacl}\\
    \mathrm{P}_{\mathrm{D}, k}&=Q\left(\frac{\theta_k-\esp\left[\lambda_k| \mathcal{H}_1\right]}{\sigma_{\lambda_k \mid \mathcal{H}_1}}\right),\label{eq:pdcl}
\end{align}
where $\sigma_{\lambda_k \mid \mathcal{H}_0}$ and $\sigma_{\lambda_k \mid \mathcal{H}_1}$ are the standard deviations of $\lambda_k$, given $ \mathcal{H}_0$, and $\lambda_k $, given $\mathcal{H}_1$, respectively. We start by deriving the expectation terms in~\eqref{eq:Pfacl} and~\eqref{eq:pdcl}. For that, knowing that the tag generated from the $k$th user is uncorrelated with the message signals and independent from other tags, and given the test statistic under hypothesis $\mathcal{H}_0$ presented in~\eqref{eq:lambdakh0}, we obtain $\esp\left[\lambda_k|\mathcal{H}_0\right] = 0$. 

Then, based on~\eqref{eq:lambdakh1}, we have  $$\esp\left[\lambda_k|\mathcal{H}_1\right] = \frac{\esp[a_{k,k}]L}{\sqrt{\rho_k}}.$$ To compute $\esp\left[a_{k,k}\right]$, the channel vector from the $k$th user to the $m$th AP is decomposed as
\begin{align}
    \mathbf{g}_{mk} = \hat{\mathbf{g}}_{mk} + \bm{\epsilon}_{mk},
\end{align}
where $\bm{\epsilon}_{mk}$ is the channel estimation error vector, which has i.i.d. $\mathcal{CN}\left(0, \frac{\beta_{mk}}{\tau_{\text{p}}\rho_k\beta_{mk}+1}\right)$ entries. Since we employ an MMSE channel estimation, $\hat{g}_{mk}$ and $\bm{\epsilon}_{mk}$ are independent~\cite{6736537}. Thus, $\esp\left[\lambda_k|\mathcal{H}_1\right]$ is rewritten as
\begin{align}\label{eq:esph1}
    \esp\left[\lambda_k|\mathcal{H}_1\right] &= \esp\left[\sqrt{\rho_k}\sum_{m=1}^{M}\left(\mathbf{w}_{mk}^H\hat{\mathbf{g}}_{mk} + \mathbf{w}_{mk}^H\bm{\epsilon}_{mk}\right)\right] \frac{L}{\sqrt{\rho_k}}\nonumber\\
    &= ML.
\end{align}
Subsequently, since $\lambda_k$ in \eqref{eq:test} is given by a sum of independent terms, and knowing that $$\var(\ttag_k^H\mathbf{s}_k)=\var(\ttag_k^H\mathbf{s}_{k'})=\var(\ttag_k^H\mathbf{s}_e) = L,$$ the variance of $\lambda_k$, given $\mathcal{H}_0$, is obtained as

%~

\begin{align}\label{eq:varh0}
    &\var(\lambda_k|\mathcal{H}_0) = L\left(\frac{\rho_{\text{s}}^2}{\rho_k\rho_{\text{t}}^2}(\sigma_{k,k}^2+(\sqrt{\rho_k}M-1)^2) \right.\nonumber\\
    &\qquad\qquad\quad\left.+ \left(\frac{1}{\rho_k} + \frac{\rho_{\text{s}}^2}{\rho_k\rho_{\text{t}}^2}\right)\left(\sum_{k' \neq k}^K\sigma_{k,k'}^2 + \sigma_{k,e}^2\right)\right.\nonumber\\ 
   & \qquad\qquad\quad\left.+ \frac{1}{\rho_k\rho_{\text{t}}^2}\mathrm{Tr}(\mathbf{C}_{nk})\right).
\end{align}

From~\eqref{eq:lambdakh1}, we note that the variance of $\lambda_k|\mathcal{H}_1$ is also given by \eqref{eq:varh0}.
% \begin{align}\label{eq:varh1}
%     \var(\lambda_k \mid \mathcal{H}_1) &= L\left(\sigma_{k,k}^2 + \frac{\rho_{\text{s}}^2}{\rho_{\text{t}}^2}(\sigma_{k,k}^2+(\sqrt{\rho_{\text{t}}}M-1)^2)\right.\nonumber\\
%     &\left. + \left(1 + \frac{\rho_{\text{s}}^2}{\rho_{\text{t}}^2}\right)\left(\sum_{k' \neq k}^K\sigma_{k,k'}^2 + \sigma_{k,e}^2\right) + \frac{1}{\rho_{\text{t}}^2}\sigma_{nk}^2\right).
% \end{align}
Then, by taking the square-root of \eqref{eq:varh0}, and replacing it alongside $\esp\left[\lambda_k|\mathcal{H}_0\right]$ into \eqref{eq:Pfacl}, \eqref{eq:pfakf} is achieved. Moreover, \eqref{eq:pdkf} is attained by plugging \eqref{eq:esph1}, and the square-root of \eqref{eq:varh0} into \eqref{eq:pdcl}. From the Neyman-Pearson Theorem~\cite{850674}, the optimal decision threshold, $\theta_k^*$, is calculated by setting $\mathrm{P}_{\mathrm{FA}, k} = p_{\mathrm{FA},k}$ as in \eqref{eq:opttheta}. 

%To compute $\var(\lambda_k \mid \mathcal{H}_0)$

\bibliographystyle{IEEEtran}  
\bibliography{references}

% Generated by IEEEtran.bst, version: 1.14 (2015/08/26)
\begin{thebibliography}{10}
\providecommand{\url}[1]{#1}
\csname url@samestyle\endcsname
\providecommand{\newblock}{\relax}
\providecommand{\bibinfo}[2]{#2}
\providecommand{\BIBentrySTDinterwordspacing}{\spaceskip=0pt\relax}
\providecommand{\BIBentryALTinterwordstretchfactor}{4}
\providecommand{\BIBentryALTinterwordspacing}{\spaceskip=\fontdimen2\font plus
\BIBentryALTinterwordstretchfactor\fontdimen3\font minus \fontdimen4\font\relax}
\providecommand{\BIBforeignlanguage}[2]{{%
\expandafter\ifx\csname l@#1\endcsname\relax
\typeout{** WARNING: IEEEtran.bst: No hyphenation pattern has been}%
\typeout{** loaded for the language `#1'. Using the pattern for}%
\typeout{** the default language instead.}%
\else
\language=\csname l@#1\endcsname
\fi
#2}}
\providecommand{\BIBdecl}{\relax}
\BIBdecl

\bibitem{9279294}
N.~Xie, Z.~Li, and H.~Tan, ``A survey of physical-layer authentication in wireless communications,'' \emph{IEEE Commun. Surv. Tutor.}, vol.~23, no.~1, pp. 282--310, Firstquarter 2021.

\bibitem{Mobini:TIFS:2019}
Z.~Mobini, M.~Mohammadi, and C.~Tellambura, ``Wireless-powered full-duplex relay and friendly jamming for secure cooperative communications,'' \emph{IEEE Trans. Inf. Forensics Security}, vol.~14, no.~3, pp. 621--634, Mar. 2019.

\bibitem{8920228}
P.~Zhang, T.~Taleb, X.~Jiang, and B.~Wu, ``Physical layer authentication for massive {MIMO} systems with hardware impairments,'' \emph{IEEE Trans. Wireless Commun.}, vol.~19, no.~3, pp. 1563--1576, Mar. 2020.

\bibitem{10694672}
M.~Srinivasan, L.~Senigagliesi, H.~Chen, A.~Chorti, M.~Baldi, and H.~Wymeersch, ``{AoA}-based physical layer authentication in analog arrays under impersonation attacks,'' in \emph{Proc. IEEE SPAWC}, Sep. 2024, pp. 496--500.

\bibitem{10005830}
P.~Zhang, Y.~Teng, Y.~Shen, X.~Jiang, and F.~Xiao, ``Tag-based {PHY}-layer authentication for {RIS}-assisted communication systems,'' \emph{IEEE Trans. Dependable Secure Comput.}, vol.~20, no.~6, pp. 4778--4792, Nov. 2023.

\bibitem{10031268}
N.~Xie, M.~Sha, T.~Hu, and H.~Tan, ``Multi-user physical-layer authentication and classification,'' \emph{IEEE Trans. Wireless Commun.}, vol.~22, no.~9, pp. 6171--6184, Sep. 2023.

\bibitem{9672766}
S.~J. Maeng, Y.~Yapıcı, {\.I}.~G{\"u}ven{\c c}, A.~Bhuyan, and H.~Dai, ``Precoder design for physical-layer security and authentication in massive {MIMO} {UAV} communications,'' \emph{IEEE Trans. on Veh. Technol.}, vol.~71, no.~3, pp. 2949--2964, Mar. 2022.

\bibitem{5740593}
P.~L. Yu and B.~M. Sadler, ``{MIMO} authentication via deliberate fingerprinting at the physical layer,'' \emph{IEEE Trans. Inf. Forensics and Secur.}, vol.~6, no.~3, pp. 606--615, Sep. 2011.

\bibitem{Mohammadi:PROC.2024}
M.~Mohammadi, Z.~Mobini, H.~Ngo, and M.~Matthaiou, ``Next-generation multiple access with cell-free massive {MIMO},'' \emph{Proc. {IEEE}}, vol. 112, no.~9, pp. 1372--1420, Sept. 2024.

\bibitem{8546764}
N.~Xie, C.~Chen, and Z.~Ming, ``Security model of authentication at the physical layer and performance analysis over fading channels,'' \emph{IEEE Trans. Dependable Secure Comput.}, vol.~18, no.~1, pp. 253--268, Jan. 2021.

\bibitem{bookCrypto}
A.~J. Menezes, P.~C. Van~Oorschot, and S.~A. Vanstone, \emph{Handbook of Applied Cryptography}.\hskip 1em plus 0.5em minus 0.4em\relax CRC press, 2018.

\bibitem{4451099}
P.~L. Yu, J.~S. Baras, and B.~M. Sadler, ``Physical-layer authentication,'' \emph{IEEE Trans. Inf. Forensics and Secur.}, vol.~3, no.~1, pp. 38--51, Mar. 2008.

\bibitem{7120016}
P.~L. Yu, G.~Verma, and B.~M. Sadler, ``Wireless physical layer authentication via fingerprint embedding,'' \emph{IEEE Commun. Mag.}, vol.~53, no.~6, pp. 48--53, Jun. 2015.

\bibitem{850674}
U.~Maurer, ``Authentication theory and hypothesis testing,'' \emph{IEEE Trans. Inf. Theory}, vol.~46, no.~4, pp. 1350--1356, July 2000.

\bibitem{7827017}
H.~Q. Ngo, A.~Ashikhmin, H.~Yang, E.~G. Larsson, and T.~L. Marzetta, ``Cell-free massive {MIMO} versus small cells,'' \emph{IEEE Trans. Wireless Commun.}, vol.~16, no.~3, pp. 1834--1850, Jan. 2017.

\bibitem{6736537}
H.~Q. Ngo, E.~G. Larsson, and T.~L. Marzetta, ``Massive {MU-MIMO} downlink {TDD} systems with linear precoding and downlink pilots,'' in \emph{Proc. IEEE ALLERTON}, Oct. 2013, pp. 293--298.

\end{thebibliography}

\end{document}